\documentclass[12pt%,reqno
]{article}
\usepackage{fullpage}
\usepackage{amssymb}
\usepackage{authblk}
\usepackage{latexsym}%
\usepackage{amsmath}
\usepackage{amsfonts}%
\usepackage{graphicx}%
\usepackage[normalem]{ulem}%
\usepackage{array}%
\usepackage{amsthm}%
\usepackage{color}\usepackage[dvipsnames]{xcolor}%
\usepackage{enumitem}

\usepackage{mathtools, amsmath, amsthm, amssymb, mathtools}

\usepackage{pgfkeys}%
\usepackage{fancyheadings}
\usepackage[colorlinks=true,
linkcolor=webgreen,
filecolor=webbrown,
citecolor=webgreen,
bookmarks=false
]{hyperref}
\definecolor{webgreen}{rgb}{0,.5,0}
\definecolor{webbrown}{rgb}{.6,0,0}

\usepackage{float}%
\usepackage{tikz}%,tikz-cd}%\usetikzlibrary{matrix,arrows}\usetikzlibrary{shapes.geometric, backgrounds, calc,
\usetikzlibrary{patterns}
\usetikzlibrary{decorations}
\usetikzlibrary{decorations.pathreplacing}

\newcommand{\comm}[1]{#1}
\newcommand{\tb}{%
\begin{tikzpicture}[scale=0.2, inner sep=-1pt]
   \tikzstyle{br} = [line width=1.3pt];
   \scriptsize
  \node (n1) at (2,3) {\tiny $1$};
   \node (n2) at (1,1) {\tiny $2$};
   \node (n3) at (3,1) {\tiny $3$};
   \path[br]
   (1.75,2.62) edge (1,1.5)
   (2.25,2.62) edge (3,1.5)
   ;
\end{tikzpicture}
}
\newcommand{\tc}{%
\begin{tikzpicture}[scale=0.2, inner sep=-1pt]
   \tikzstyle{br} = [line width=1.3pt];
   \scriptsize
  \node (n1) at (2,3) {\tiny $1$};
   \node (n2) at (1,1) {\tiny $3$};
   \node (n3) at (3,1) {\tiny $2$};
   \path[br]
   (1.75,2.62) edge (1,1.5)
   (2.25,2.62) edge (3,1.5)
   ;
\end{tikzpicture}
}

\newcommand{\ttb}{%
\begin{tikzpicture}[scale=0.2, inner sep=-1pt]
   \tikzstyle{br} = [line width=1.3pt];
   \scriptsize
  \node (n1) at (2,3) {$\bullet$};
   \node (n2) at (1,2) {$\bullet$};
   \node (n3) at (3,1) {$\bullet$};
   \path[br]
   (n1) edge (n2)
   (n1) edge (n3)
   ;
\end{tikzpicture}
}
\newcommand{\ttc}{%
\begin{tikzpicture}[scale=0.2, inner sep=-1pt]
   \tikzstyle{br} = [line width=1.3pt];
   \scriptsize
  \node (n1) at (2,3) {$\bullet$};
   \node (n2) at (1,1) {$\bullet$};
   \node (n3) at (3,2) {$\bullet$};
   \path[br]
   (n1) edge (n2)
   (n1) edge (n3)
   ;
\end{tikzpicture}
}

\newcommand{\ma}{%
\begin{tikzpicture}[scale=0.2, inner sep=-1pt]
   \tikzstyle{br} = [line width=1.3pt];
   \scriptsize
  \node (n1) at (2,2) {$\bullet$};
   \node (n2) at (1,1) {$\bullet$};
   \path[br]
   (n1) edge (n2)
   ;
   \end{tikzpicture}
}

\newcommand{\maa}{%
\begin{tikzpicture}[scale=0.2, inner sep=-1pt]
   \tikzstyle{br} = [line width=1.3pt];
   \scriptsize
  \node (n1) at (1,2) {$\bullet$};
   \node (n2) at (2,1) {$\bullet$};
   \path[br]
   (n1) edge (n2)
   ;
   \end{tikzpicture}
}

\newcommand{\mooo}{%
\begin{tikzpicture}[scale=0.2, inner sep=-1pt]
   \tikzstyle{br} = [line width=1.3pt];
   \scriptsize
  \node (n1) at (2,3) {$\bullet$};
   \node (n2) at (1,2) {$\bullet$};
   \node (n3) at (3,1) {$\bullet$};
   \path[br]
   (n1) edge (n2)
   (n1) edge (n3)
   ;
\end{tikzpicture}
}
\newcommand{\mo}{%
\begin{tikzpicture}[scale=0.2, inner sep=-1pt]

   \tikzstyle{br} = [line width=1.3pt];
   \scriptsize

   \node (n1) at (3,3) {$\bullet$};
   \node (n2) at (2,2) {$\bullet$};
   \node (n3) at (1,1) {$\bullet$};

   \path[br]
   (n1) edge (n2)
   (n2) edge (n3)
   ;
\end{tikzpicture}
}

\newcommand{\moool}{%
\begin{tikzpicture}[scale=0.2, inner sep=-1pt]
   \tikzstyle{br} = [line width=1.3pt];
   \scriptsize
  \node (n1) at (2,3) {$\bullet$};
   \node (n2) at (3,2) {$\bullet$};
   \node (n3) at (1,1) {$\bullet$};
   \path[br]
   (n1) edge (n2)
   (n1) edge (n3)
   ;
\end{tikzpicture}
}
\newcommand{\moo}{%
\begin{tikzpicture}[scale=0.2, inner sep=-1pt]
   \tikzstyle{br} = [line width=1.3pt];
   \scriptsize
   \node (n1) at (3,3) {$\bullet$};
   \node (n2) at (1,2) {$\bullet$};
   \node (n3) at (2,1) {$\bullet$};
   \path[br]
   (n1) edge (n2)
   (n2) edge (n3)
   ;
\end{tikzpicture}
}
\newcommand{\mol}{%
\begin{tikzpicture}[scale=0.2, inner sep=-1pt]

   \tikzstyle{br} = [line width=1.3pt];
   \scriptsize

   \node (n1) at (1,3) {$\bullet$};
   \node (n2) at (2,2) {$\bullet$};
   \node (n3) at (3,1) {$\bullet$};

   \path[br]
   (n1) edge (n2)
   (n2) edge (n3)
   ;
\end{tikzpicture}
}

\newcommand{\mool}{%
\begin{tikzpicture}[scale=0.2, inner sep=-1pt]
   \tikzstyle{br} = [line width=1.3pt];
   \scriptsize
   \node (n1) at (1,3) {$\bullet$};
   \node (n2) at (3,2) {$\bullet$};
   \node (n3) at (2,1) {$\bullet$};
   \path[br]
   (n1) edge (n2)
   (n2) edge (n3)
   ;
\end{tikzpicture}
}

\newcommand{\seqnum}[1]{\href{http://oeis.org/#1}{\underline{#1}}}

\theoremstyle{plain}

\newtheorem{cor}{Corollary}
\newtheorem{thm}{Theorem}

\theoremstyle{definition}

\title{Patterns in treeshelves}

\author{Jean-Luc Baril, Sergey Kirgizov and Vincent Vajnovszki}
\affil{LE2I UMR CNRS, Universit\'e Bourgogne Franche-Comt\'e\\
21078 Dijon, France
\authorcr \texttt{\{barjl\}\{sergey.kirgizov\}\{vvajnov\}@u-bourgogne.fr}}

\begin{document}

\maketitle

\begin{abstract}  We study the distribution and the popularity
of left children on sets of treeshelves avoiding a pattern of size three.
(Treeshelves are ordered binary increasing trees where every child is connected
to its parent by a left or a right link.)
The considered patterns are sub-treeshelves, and for each such a pattern
we provide exponential generating function for the corresponding distribution and popularity.
Finally, we present constructive bijections between treeshelves avoiding a pattern of size three
and some classes of simpler combinatorial objects.
\end{abstract}


\section{Introduction and notation}

The study of patterns in permutations was first introduced by Knuth \cite{knu73},
and continues to be an active area of research today.
Recently, patterns have been studied in contexts other than permutations,
see for instance \cite{Corteel_and_co,Mansour_Shattuck}
where the combinatorial class under consideration are inversion sequences, which
can be seen as an alternative representation for permutations.
The present paper deals with treeshelves (formally defined below) which
are still another class in bijection with
permutations, and patterns  are sub-treeshelves contained or avoided in a similar way as
consecutive patterns do in permutations or in inversion sequences.
More precisely, we consider the class of unrestricted treeshelves
and of those avoiding a pattern of size 3 (treeshelves avoiding a
pattern of size 2 collapse trivially to a singleton set).
We not only enumerate these classes for any avoider of size 3,
but also give bivariate generating functions  with respect to the size and to the
number of occurrences of a second pattern of size 2.
As a byproduct we obtain the popularity among these classes of the pattern of size 2,
obtaining counting sequences which are not yet recorded in Sloane's Encyclopedia of Integer
Sequences \cite{sloa}.

Treeshelves are particular classes of binary increasing trees,
considered for example in Fran{\c{c}}on's
work~\cite{fra} in the context of data structures for binary search methods.
An {\em increasing tree} of size $n$,  is a rooted tree with $n$ nodes
labeled by distinct integers in $\{1,2,\dots,n\}$, so that the sequences of labels are increasing along all branches starting at the root (and thus,
the root is labeled by $1$).
A {\em binary increasing tree} (sometimes called 0-1-2 increasing tree)
is an increasing tree where every node has at most two children.
Many studies ({\it e.g.}, \cite{ber,bod,cal,don,kut}) investigate binary increasing trees, but very few deal with such trees endowed with the additional property that every child
(including those with no siblings) is connected to its parent by either a left or a right link.
We call such a binary increasing tree {\it treeshelf} (or {\it t-shelf} for short), and its size is the number of its nodes, see Figure \ref{tre} for a size 7 t-shelf.
We denote by $\mathcal{B}_n$ the set of size $n$ t-shelves, and $\mathcal{B}_1$ consists
of the single one-node t-shelf.
Often it is more convenient to represent graphically t-shelves by trees where the
integers labeling the nodes are proportional with the lengths of the branches.
For example, the size 3 t-shelf

\begin{center}
$\tb$ is represented by $\ttb$\,, and $\tc$ is represented by $\ttc$\,,
\end{center}
see also Figure~\ref{tre}.
In this representation, $\mathcal{B}_2=\{\ma\, ,\, \maa\}$,
and $\mathcal{B}_3=\{
\mo\, ,\,
\moool\, ,\,
\moo\, ,\,
\mooo\, ,\,
\mool\, ,\,
\mol\, ,\,
\}$.

We denote $\cup_{n\geq 0} \mathcal{B}_n$ by $\mathcal{B}$, and  $\cup_{n\geq 1} \mathcal{B}_n$
by $\mathcal{B}^\bullet$.
The labeled tree rooted at the left child of the root of a t-shelf $T$ becomes
a t-shelf after appropriately relabeling its nodes, and in the following
we refer to it as the {\it left t-shelf of $T$}, and similarly for
the {\it right t-shelf of $T$}.

There is a bijection between $\mathcal{B}_n$ and the set of permutations of size $n$, and so the cardinality of $\mathcal{B}_n$ is $n!$.
Indeed, to any t-shelf $T$ in $\mathcal{B}_n$ we can uniquely associate the length $n$
permutation $\pi=\alpha (n-r(T)+1) \beta$, where $r(T)$ is the label of the root of $T$, and $\alpha$  (resp. $\beta$) is recursively defined from the left (resp. right) t-shelf of $T$ (see again Figure \ref{tre}).
As mentioned by Bergeron, Flajolet, and Salvy~\cite{ber}, this construction appears in~\cite{fra}
and thereafter recalled in Stanley's book~\cite{sta}. Additional  information (including historical
notes) about binary and other families of increasing trees can be
found for example in~\cite{ber,com,gou}.

\begin{figure}[ht!]
\centering
\comm{
\begin{tikzpicture}[scale=0.45, inner sep=0pt]
  \large
  \draw[step=1cm,gray, very thin, xshift=0.5cm, yshift=0.5cm] (0, 0) grid (7,7);
  \node (n5) at (1,5) {$\bullet$};
  \node (n3) at (2,3) {$\bullet$};
  \node (n7) at (3,7) {$\bullet$};

  \node (n4) at (4,4) {$\bullet$};
  \node (n6) at (5,6) {$\bullet$};
  \node (n2) at (6,2) {$\bullet$};

  \node (n1) at (7,1) {$\bullet$};

  \scriptsize
  \node (C5) at (1,-0.2) {5};
  \node (C3) at (2,-0.2) {3};
  \node (C7) at (3,-0.2) {7};
  \node (C4) at (4,-0.2) {4};
  \node (C6) at (5,-0.2) {6};
  \node (C2) at (6,-0.2) {2};
  \node (C1) at (7,-0.2) {1};

  \path[line width=1.3pt]
  (n7) edge (n6) %edge (n5)

%  (n5) edge (n3)

  %(n6) edge (n4) edge (n2)

  (n2) edge (n1)
  ;

  \path[line width=1pt, dashed, color = red]

  (n6) edge (n4) edge (n2)
  ;

  \path[line width=2pt, dotted, color = green]

  (n7) edge (n5)
  (n5) edge (n3)
  ;

\end{tikzpicture}\qquad\begin{tikzpicture}[scale=0.45, inner sep=0pt]
  \large
  %\draw[step=1cm,gray, very thin, xshift=0.5cm, yshift=0.5cm] (0, 0) grid (7,7);
  \node (n8) at (0,0.5) {};
  \node (n5) at (1,6) {$3$};
  \node (n3) at (2,4) {$5$};
  \node (n7) at (3,8) {$1$};

  \node (n4) at (4,5) {$4$};
  \node (n6) at (5,7) {$2$};
  \node (n2) at (6,3) {$6$};

  \node (n1) at (7,2) {$7$};

  \path[line width=1.3pt]
  (n7) edge (n6) %edge (n5)

%  (n5) edge (n3)

  %(n6) edge (n4) edge (n2)

  (n2) edge (n1)
  ;

  \path[line width=1pt, dashed, color = red]

  (n6) edge (n4) edge (n2)
  ;

  \path[line width=2pt, dotted, color = green]

  (n7) edge (n5)
  (n5) edge (n3)
  ;

\end{tikzpicture}
}
\caption{The t-shelf corresponding to the permutation $5\,3\,7\,4\,6\,2\,1$;
dashed/dotted lines correspond to different patterns of size three.}
\label{tre}
\end{figure}

In this paper we are interested in the sets of t-shelves avoiding a pattern
$P\in \mathcal{B}_3$, {\it i.e.},  the sets of those that do not contain any occurrence of $P$.
The containment/avoidance of a pattern in a t-shelf can most easily be explained
with examples. The avoidance of $\mooo$ in a t-shelf $T$ means that $T$ does not contain any node where the label of its left child is less than that of its right child. The t-shelf in Figure \ref{tre} contains only one pattern $\mooo$ (illustrated by dashed lines), one pattern $\moo$ (dotted)
and avoids the pattern $\mo$.

\medskip

Since the number of $\maa\,$ patterns in a t-shelf is equal to the size of the t-shelf
minus the number of $\ma\,$ patterns, minus one, in the following we will consider only $\ma\,$ patterns.
Moreover, an occurrence of the $\ma\,$ pattern is equivalent to
that of a left child in the underlying tree
of the t-shelf, we will refer to this pattern as a left child (similarly the pattern $\maa\,$
corresponds to a right child).
Also, since the patterns $\moo\,$ and $\mool\,$ are equivalent by symmetry,
and so are the patterns $\mo\,$ and $\mol\,$, and the patterns $\mooo\,$ and $\moool\,$,
we will consider
only avoiders $P$ in $\{\moo\,,\mo\,,\mooo\,\}$.

T-shelves are labeled combinatorial objects, and so
it is appropriate to use exponential generating functions (e.g.f.) for the enumerative analysis of them.
In Section \ref{avoiding}, for each of the avoiders $P$ above mentioned,
we consider the set $\mathcal{B}(P)$ of t-shelves avoiding $P$,
or $\mathcal{B}^\bullet(P)$ when we restrict to non-empty t-shelves.
We provide a bivariate exponential generating function for each $\mathcal{B}(P)$
with respect to the size and the number of left children,
that is, function where the coefficient of $\frac{z^ny^k}{n!}$ in its series expansion is the number of t-shelves of size $n$ having
exactly $k$ left children, and deduce the e.g.f. for $\mathcal{B}(P)$ with respect to the size.
We also give the e.g.f. for the popularity of the left children among $\mathcal{B}(P)$,
function where the coefficient of  $\frac{z^n}{n!}$ in its series expansion is the total number of left children
appearing in all
size $n$ t-shelves in $\mathcal{B}(P)$. These results are summarized  in Tables 1 and 2.

Our method consists in constructing recursively the combinatorial class in question
from two smaller classes, $\mathcal{A}_1$ and $\mathcal{A}_2$,
using the usual labeled product $\mathcal{A}_1\star\mathcal{A}_2$ and the boxed product
$\mathcal{A}_1^\square\star\mathcal{A}_2$.
The boxed product $\mathcal{A}_1^\square\star\mathcal{A}_2$ is a subset of $\mathcal{A}_1\star\mathcal{A}_2$ where the smallest label appears in the $\mathcal{A}_1$ component. See~\cite{fla} for more information about the boxed product and its application on labeled combinatorial structures.

Theorems \ref{new_th_bij}-\ref{th6} in Section \ref{bijections} give constructive proofs of some
results in Section \ref{avoiding}, namely constructive bijections between:
(i) t-shelves avoiding $\moo\,$ and set partitions,
(ii) (unordered) binary increasing trees where every node of degree one has either a left or a
right child and t-shelves avoiding the pattern $\mooo$, and
(iii) unordered binary increasing trees and t-shelves avoiding the pattern $\mo$.

\section{T-shelves avoiding  a size 3 pattern}
\label{avoiding}
%

%
%or equivalently, as
%
%\begin{equation*}
%\mathcal{B}^\bullet=\epsilon+\mathcal{Z}+
%2\mathcal{Z}^\square \star \mathcal{B}
%+\mathcal{Z}^\square \star \mathcal{B}^2,
%\end{equation*}
%

We begin this section by considering unrestricted t-shelves, then we extend our
approach to those avoiding a pattern in $\{\moo\,,\mo\,,\mooo\,\}\subset\mathcal{B}_3$.

A t-shelf is either empty or consists of a root with two (possibly empty) children.
Thus, the set $\mathcal{B}$ of unrestricted t-shelves can be expressed as
\begin{equation*}
\mathcal{B} = \epsilon+
\mathcal{Z}^\square \star \mathcal{B}^2,
\label{eqB}
\end{equation*}
where $\mathcal{Z}$ corresponds to the atom, {\it i.e.}, the singleton formed by the unique object of size one.

As the boxed product $\mathcal{A}^\square_1\star\mathcal{A}_2$ has its exponential generating function given by $\int_0^z \! \partial_t A_1(t)\cdot A_2(t)\, \mathrm{d}t$,
where  $A_1(t)$ and $A_2(t)$ are the exponential generating functions of
$\mathcal{A}_1$ and $\mathcal{A}_2$, respectively
(see~\cite[Theorem II.5]{fla}), we obtain the differential equation
$$
B(z) = 1+\int_0^z \!  B^2 (t)\, \mathrm{d}t,
$$
which, with the initial condition $B(0)=1$,
gives as expected $B(z)=\frac{1}{1-z}$, the e.g.f. for the sequence $n!$.

If we are interested in the bivariate exponential generating function
$B(z,y)$ where the coefficient of $\frac{z^ny^k}{n!}$ is the number of t-shelves of size $n$ having exactly $k$ left children
(or, equivalently by symmetry, $k$ right children),
then it is more convenient to consider the set $\mathcal{B}^\bullet$ of non-empty t-shelves.
A t-shelf $T\in\mathcal{B}^\bullet$ can be in one of the following cases:
the root of $T$ either
\begin{itemize}[topsep=-0.5mm,itemsep=-1mm]
\item[$-$] has no children ($T$ is reduced to one root node), in this case
the set of such $T$ is $\mathcal{Z}$; or
\item[$-$] has only a left or only a right child, in both cases
the set of such $T$ is $\mathcal{Z}^\square \star \mathcal{B}^\bullet$; or
\item[$-$] has both left and right children,
the set of such $T$ is $\mathcal{Z}^\square \star \mathcal{B}^\bullet\star \mathcal{B}^\bullet$.
\end{itemize}
Thus, $\mathcal{B}^\bullet$ can be expressed as
\begin{equation*}
\mathcal{B}^\bullet = \mathcal{Z}+
\mathcal{Z}^\square \star \mathcal{B}^\bullet +
\mathcal{Z}^\square \star \mathcal{B}^\bullet+
\mathcal{Z}^\square \star (\mathcal{B}^\bullet)^2,
\end{equation*}
and after multiplying by $y$ whenever a new left child is created,
we obtain the differential equation
 $$B^\bullet(z,y) = z+ \int_0^z \! B^\bullet(t,y) \, \mathrm{d}t
+ y \int_0^z \! B^\bullet(t,y) \, \mathrm{d}t
 + y \int_0^z \! (B^\bullet(t,y))^2 \, \mathrm{d}t,
$$
where $B(0,y)^\bullet=0$, and its solution is
$
B^\bullet(z,y)=
\frac {1-{\rm e}^{z \left( y-1 \right) }}{{{\rm e}^{z \left( y-1 \right) }}-y}.
$
Finally,
\begin{equation*}
B(z,y)=1+B^\bullet(z,y)=\frac {1-y }{{{\rm e}^{z \left( y-1 \right) }}-y},
\end{equation*}
%s
%
and we retrieve two well known results, see \cite[Exercise 1.9]{pet}:
\begin{itemize}[topsep=-0.5mm,itemsep=-1mm]
\item[$-$] the distribution of the left children on the set $\mathcal{B}$ has the exponential
generating function $B(z,y)$, and it is given by a shift of the Eulerian numbers
      (sequence \seqnum{A008292} in OEIS \cite{sloa}); and
\item[$-$] the popularity of the left (or right) children among $\mathcal{B}$, which is  the coefficient of $\frac{z^n}{n!}$ in
$\partial_y B(z,y)|_{y=1}={\frac {{z}^{2}}{2\,{z}^{2}-4\,z+2}}$, is
given by the Lah numbers (sequence \seqnum{A001286} in OEIS \cite{sloa}).
\end{itemize}

%\vskip2cm

%The techniques used in
%Theorems~\ref{OdownOdownO.enumeration},~\ref{OdownupA.enumeration}
%could also be extended to enumerate treeshelfs that avoid any
%connected pattern.  If $p$ is a connected treeshelf pattern, then one
%may enumerate ($\mathcal{T} without p$) by considering only prefixes
%that avoid $p$ and use properly boxed product in order to write an
%integral equation.

%***********************************************************************************************************
%***********************************************************************************************************
%******************************************111111111111111111111111*****************************************
%***********************************************************************************************************
\medskip

In the following, for each t-shelf $P\in\{\moo\,,\mo\,,\mooo\,\}$ we will count the class $\mathcal{B}(P)$
(or  $\mathcal{B}^\bullet(P)$)
of t-shelves avoiding  $P$, and explore the distribution and the popularity of left children
({\it i.e}, of the pattern $\ma\,$) among each class.

\subsection{Pattern \protect\moo}
Here we consider $\mathcal{B}(P)$ with $P=\moo$,
that is, t-shelves having all they left children with no right child, we refer to Figure~\ref{bells} for an
illustration of the shape of such a t-shelf.

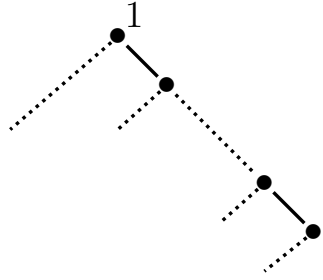
\begin{figure}[H]
\comm{
  \centering
     \begin{tikzpicture}[scale=0.65, inner sep=0pt]

  \tikzstyle{br} = [line width=1.3pt]

  \large

  \node[label={[xshift=1ex]$1$}]
        (n1) at (0, 0) {$\bullet$};
  \node (n2) at (1,-1) {$\bullet$};
  \node (n3) at (3,-3) {$\bullet$};
  \node (n4) at (4,-4) {$\bullet$};

  \draw[dotted, br] (0,0) --        (-2.2,-1.9);
  \draw[dotted, br] (1,-1) --        (0,-1.9);
  \draw[dotted, br] (3,-3) --        (2.1,-3.8);
  \draw[dotted, br] (4,-4) --        (3,-4.8);

  \path[br]
  (n1) edge (n2)
  (n2) edge[dotted] (n3)
  (n3) edge (n4)
  ;

\end{tikzpicture}
}
\caption{The shape of a t-shelf avoiding pattern \protect\moo.}
\label{bells}
\end{figure}

\begin{thm}
  \label{th1}
Let $P$ be the pattern $\moo$. The bivariate e.g.f. for $\mathcal{B}(P)$ with respect to the size of t-shelves and the number of left children is given by
$$
C(z,y) = e^\frac{e^{zy}-1}{y}.
$$
\end{thm}
\proof
Let $\mathcal{C}=\mathcal{B}(P)$ and $T\in \mathcal{C}$.
According to the shape of the t-shelves in $\mathcal{C}$ (see Figure~\ref{bells}),
if $T$ is non-empty, then it is obtained by a pair of t-shelves, namely
a non-empty t-shelf with no right children containing the smallest label of $T$, and a second unrestricted
t-shelf in $\mathcal{C}$. The set of such non-empty t-shelves is
$\mathcal{D}^\square\star\mathcal{C}$, where $\mathcal{D}$ is the set of
non-empty t-shelves with no right children. Thus we have
$$\mathcal{C}=\epsilon+\mathcal{D}^\square\star\mathcal{C}.
$$

Since the bivariate exponential generating function for $\mathcal{D}$ is $D(z,y)= \frac{e^{zy} - 1}{y}$, we obtain the differential equation
$$ C(z,y) = 1 +
\int_0^z \! e^{ty}\cdot C(t,y)\cdot \, \mathrm{d}t
$$
where $C(0,y)=1$, with the solution $C(z,y) = e^\frac{e^{zy}-1}{y}.$
\endproof

By calculating $C(z,1)$ we have the following corollary.
\begin{cor}\label{cor2} The exponential generating function for the set $\mathcal{B}(P)$ with respect to the size of t-shelves
is $\mbox{Bell}(z)=e^{e^{z}-1}$, which generates the Bell numbers (sequence {\em \seqnum{A000110}} in OEIS \cite{sloa}).
\end{cor}
\begin{cor}\label{cor3}
The popularity of the left children among the set $\mathcal{B}(P)$ is given by the exponential generating
function $$PC(z)=\left( z{{\rm e}^{z}}-{{\rm e}^{z}}+1 \right) {{\rm e}^{{{\rm e}^{z}}
-1}}.$$
Moreover, the coefficient $pc_n$ of $\frac{z^n}{n!}$ in $PC(z)$ satisfies $pc_n=(n+1)b_n-b_{n+1}$ where $b_n$ is the
$n${\em th} Bell number. The asymptotic of $pc_n$ is given by
$$\sqrt {n} \left(\frac {n}{W(n)}\right) ^{n+\frac{1}{2}}{{\rm e}^{ \frac {n}{W(n)}-n-1} },$$
where $W$ is the Lambert function \cite{euler,Polya_Szego}, that is,
$W\left( n \right)$ is the unique solution of \mbox{$W(n)\cdot e^{W(n)}=n$}.
\end{cor}
\noindent
(The first terms of $pc_n$, $n\geq 2$, are $1,5,23,109,544,2876,16113,95495$.)
\proof The popularity is given by $\partial_y C(z,y)|_{y=1}=\left( z{{\rm e}^{z}}-{{\rm e}^{z}}+1 \right) {{\rm e}^{{{\rm e}^{z}}
-1}}$.
The recurrence for $pc_n$ is directly obtain from the
relation $(z-1)\partial_z \mbox{Bell}(z)+\mbox{Bell}(z)=PC(z)$.

\noindent
Finally, the asymptotic follows from the asymptotic formula due to M. Klazar
\cite[Proposition 2.6]{kla} and D.E. Knuth \cite[eq. (30), p. 69]{knu}:
$$\frac{b_{n+1}}{b_n}\sim\frac{n}{\ln (n)},
$$
and from the well known asymptotic for the Bell numbers
(see A.M. Odlyzko \cite{odl}):
$$ {\frac {1}{\sqrt {n}}}\left( {\frac {n}{{\it W} \left( n \right) }} \right) ^{n+\frac{1}{2}}
{{\rm e}^{{\frac {n}{{\it W} \left( n \right) }}-n-1}},
$$
where
$\it W \left( n \right)$ is the unique solution of $ {\it W}\left( n \right)\cdot e^{{\it W}\left( n \right)}=n$.
\endproof

%***********************************************************************************************************
%***********************************************************************************************************
%*******************************************2222222222222222222222******************************************
%***********************************************************************************************************

\subsection{Pattern \protect\mo}
Here we consider the set $\mathcal{B}(P)$ of t-shelves avoiding the pattern $P=\mo$.

\begin{thm}
\label{th2}
Let $P$ be the pattern $\mo$. Then the bivariate e.g.f. for $\mathcal{B}(P)$ with respect to the size of t-shelves and the number of left children is given by
 $$ E(z,y) = \frac{2 y - 1}{y \cosh{ \left (z \sqrt{- 2 y + 1} + \ln{\left
(\frac{1}{y} \left(y + \sqrt{- 2 y + 1} - 1\right) \right )} \right )}
+ y}.$$
\end{thm}

\proof
Let $\mathcal{E}=\mathcal{B}(P)$ and $T\in \mathcal{E}$.
One of the following cases can occur.
\begin{itemize}[topsep=-0.5mm,itemsep=-1mm]
\item[$-$]   $T$ is empty.
\item[$-$]  $T$ is not empty, and its root does not have a left child. In this case,
             the right t-shelf of $T$ belongs to $\mathcal{E}$ and the set of such t-shelves $T$ is
             $\mathcal{Z}^\square \star \mathcal{E}$.
\item[$-$] The root of $T$ has a left child. In this case $T$ is obtained from a pair of t-shelves satisfying the second point above, namely one formed by the root of $T$ together
             with its right t-shelf, and the other one being the left t-shelf of $T$.
             See Figure~\ref{fig2} for an illustration of this case. So,
             $T$ is the product of two t-shelves satisfying the second point above and,
             with the smallest label belonging to the first t-shelf.
             Thus, the set of such t-shelves $T$ is $\mathcal{F}^\square \star \mathcal{F}$ where $\mathcal{F}=\mathcal{Z}^\square \star \mathcal{E}$.
\end{itemize}
Combining these cases we have
$$\mathcal{E} = \mathcal{\epsilon} + \mathcal{Z}^\square \star \mathcal{E} +
  \left( \mathcal{Z}^\square \star \mathcal{E} \right)^\square
  \star \left( \mathcal{Z}^\square \star \mathcal{E} \right),
$$
which yields the differential equation
  $$  E(z,y)  = 1 + \int_0^z \! E(t,y) \mathrm{d}t +
                                          y\cdot\int_0^z \!
                                            \left(
                                                  E(u,y)\cdot  \int_0^u \! E(t,y) \, \mathrm{d}t
                                            \right) \, \mathrm{d}u,$$
with the initial conditions $E(0,y)=1$ and $\partial_zE(z,y)|_{z=0}=1.$
A simple calculation (using Maple for instance) gives the desired result.
\endproof

  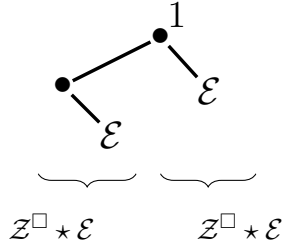
\begin{figure}[H]
\comm{
    \centering
    \begin{tikzpicture}[scale=0.65, inner sep=0pt]
  \tikzstyle{br} = [line width=1.3pt]
 \large
 \node[label={[xshift=1ex]$1$}]
        (n1) at (0,0) {$\bullet$};
  \node (n2) at (-2, -1) {$\bullet$};
  \node (A1) at ( 1, -1.1) {$\mathcal{E}$};
  \node (A2) at (-1, -2) {$\mathcal{E}$};
  \path[br]
       (n1) edge (A1) edge (n2)
       (n2) edge (A2);
   \draw [decorate,
          decoration={brace,amplitude=5pt}]
          (1.95,-2.8) -- (0.0,-2.8) node [below, xshift=30pt, yshift=-15pt]
         {\small$\mathcal{Z}^\square \star \mathcal{E}$};
   \draw [decorate,
          decoration={brace,amplitude=5pt}]
          (-0.5,-2.8) -- (-2.5,-2.8) node [below, xshift=5pt, yshift=-15pt]
         {\small$\mathcal{Z}^\square \star \mathcal{E}$};
\end{tikzpicture}
}
\caption{Illustration of a t-shelf satisfying the third case in the proof of Theorem \ref{th2}.}
\label{fig2}

  \end{figure}

The next corollary is obtained by calculating $E(z,1)$.
\begin{cor}
\label{cor21}
The exponential generating function for the set $\mathcal{B}(P)$ with respect to the size of trees is given by
$$ \mbox{Eul}(z)=\frac{1}{1-\sin z},$$
which yields a shift of the Euler numbers (sequence {\em \seqnum{A000111}} in OEIS \cite{sloa}--
not to be confused with Eulerian numbers).
\end{cor}

\begin{cor}
\label{cor22} The popularity $pe_n$ of the left children among the set $\mathcal{B}(P)$ is given by the exponential generating function
$$ PE(z)=
{\frac {-\sin z +1+ \left( z-1 \right) \cos z
}{ \left( 1-\sin z \right) ^{2}}
}.
$$
Moreover, the coefficient $pe_n$ of $\frac{z^n}{n!}$ satisfies $pe_n=(n+1)e_n-e_{n+1}$ where $e_n$ is the shifted Euler number defined by the e.g.f. $\mbox{Eul}(z)=\frac{1}{1-\sin(z)}$.
 The asymptotic of $pe_n$ is given by  $$\frac{8 (\pi - 2)}{\pi^3} n^2 \left( \frac{2}{\pi} \right)^n.$$
\end{cor}
\noindent
(The first terms of $pe_n$, $n\geq 2$, are $1,4,19,94,519,3144,20903,151418$).

\begin{proof} Using Theorem~\ref{th2}, $PE(z)$ is obtained by calculating $\partial_y E(z,y)|_{y=1}$. The recurrence relation is directly obtained with the relation
$PE(z)=(z-1)\partial_z \mbox{Eul}(z)+\mbox{Eul}(z)$, and
the asymptotic follows from the classical singularity analysis (see for instance \cite{fla}).
\end{proof}

%***********************************************************************************************************
%***********************************************************************************************************
%**********************************************333333333333333333*******************************************
%***********************************************************************************************************

\subsection{Pattern \protect\mooo}

We conclude this section by considering the pattern $P=\mooo$ and the set $\mathcal{B}^\bullet(P)$ of
non-empty t-shelves avoiding $P$.

\begin{thm}
\label{th3}
Let $P$ be the pattern $\mooo$. Then the bivariate e.g.f. for
$\mathcal{B}^{\bullet}(P)$ with respect to the size of t-shelves and the number of left children is given by
  $$G(z,y)=\frac{-2}{1+y-\sqrt {{y}^{2}+1}\coth \left( \frac{z\sqrt {{y}^{2}+1}}{2}
 \right)}.$$
\end{thm}

\begin{proof}
For $P=\mooo$, a non-empty t-shelf $T$ in $\mathcal{G}=\mathcal{B}^{\bullet}(P)$ is in one of the following cases:

\begin{itemize}[topsep=-0.5mm,itemsep=-1mm]
\item[$-$]   $T$ is reduced to one (root) node.
\item[$-$]  $T$ has at least two nodes and its root does not have a left child. In this case
             the set of such t-shelves $T$ is $\mathcal{Z}^\square \star \mathcal{G}$.
\item[$-$] $T$ has at least two nodes and its root does not have a right child. As above,
             the set of such t-shelves $T$ is $\mathcal{Z}^\square \star \mathcal{G}$.
\item[$-$]  The root of $T$ has both left and right children (see Figure~\ref{fig3}).
             In this case $T$ is obtained from a pair of t-shelves in $\mathcal{G}$ connected by a common root,
             with the smallest label of $T$ in its right t-shelf. The set of such $T$ is
             $\mathcal{Z}^\square \star (\mathcal{G}^\square\star\mathcal{G})$.
\end{itemize}
Combining these four cases we obtain
 $$  \mathcal{G}= \mathcal{Z} +  \mathcal{Z}^\square \star \mathcal{G}+\mathcal{Z}^\square \star \mathcal{G} +  \mathcal{Z}^\square \star
                                     \left( \mathcal{G}^\square
                                                                       \star \mathcal{G}
                                                                \right)$$
which induces the differential equation
$$
G(z,y) = z +  \int_0^z \! G(t,y) \, \mathrm{d}t +y\cdot\int_0^z \! G(t,y) \, \mathrm{d}t+
  y\cdot\int_0^z \! \int_0^u \!   \partial_t G(t,y)\cdot G(t,y)  \, \mathrm{d}t \, \mathrm{d}u,
  $$
  with the initial conditions $G(0,y)=0$ and $\partial_z G(z,y)|_{z=0}=1$.
  A simple calculation (using Maple for instance) gives the desired result.
  \end{proof}

\medskip

  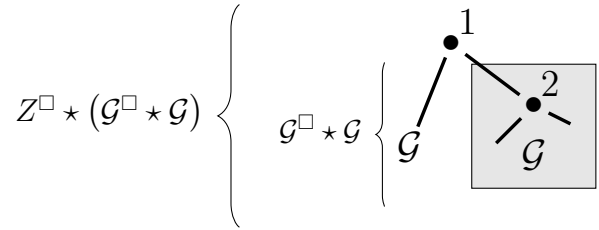
\begin{figure}[H]
\comm{
    \centering
    \begin{tikzpicture}[scale=0.55, inner sep=0pt, shorten >= 2pt, shorten <= 2pt]
  \tikzstyle{br} = [line width=1.3pt]
  \large
 % rectangle for B2
  \draw[fill=black!10!white]
       (0.5,-0.5) rectangle (3.5,-3.5);
  \node[label={[xshift=1ex]$1$}]
        (n1) at (0,0) {$\bullet$};
  \node[label={[xshift=1ex]$2$}]
        (n2) at (2, -1.5) {$\bullet$};
  \node (B1) at ( -1, -2.5) {$\mathcal{G}$};
  \node (B2) at (2, -2.7) {$\mathcal{G}$};
   \path[br]
        (n1) edge (B1) edge (n2)
        (n2) edge (1,-2.5)
        (n2) edge (3,-2);
    \draw [decorate,
           decoration={brace,amplitude=5pt}]
           (-1.5,-4) -- (-1.5,-0.4) node [left, xshift=-10pt, yshift=-26pt]
          {\small$\mathcal{G}^\square \star \mathcal{G}$};
    \draw [decorate,
           decoration={brace,amplitude=10pt}]
           (-5,-4.5) -- (-5, 1) node [left, xshift=-15pt, yshift=-42pt]
          {\normalsize$Z^\square \star \left( \mathcal{G}^\square \star \mathcal{G} \right)$};
\end{tikzpicture}
}
    \caption{Illustration of a t-shelf satisfying fourth case in the proof of Theorem~\ref{th3}.}
    \label{fig3}
  \end{figure}

\begin{cor}
  \label{cor31}
  The exponential generating function for the set $\mathcal{B}(P)$ with respect to the size of t-shelves is given by
  $$1+\frac{-2}{-\sqrt {2}\coth \left( \frac{z}{\sqrt {2}} \right) +2 }$$
  which generates the sequence {\rm \seqnum{A131178}} in OEIS \cite{sloa}.
\end{cor}

\begin{cor}
  \label{cor32} The popularity  $pg_n$ of the left children among the set $\mathcal{B}(P)$ is given by the exponential generating function
  $$ PG(z)={\frac {{e}^{\sqrt {2}z} \left( 4\,z-4 \right) - \left( \sqrt {2}-2
 \right) {e}^{2\,\sqrt {2}z}+\sqrt {2}+2}{ \left(  \left( \sqrt {2}-2
 \right) {e}^{\sqrt {2}z}+2+\sqrt {2} \right) ^{2}}}.
$$
Moreover, the asymptotic of the coefficient $pg_n$ of $\frac{z^n}{n!}$ is given by $$
n \left(\frac{\sqrt{2}}{\ln{\left (2 \sqrt{2} + 3 \right )}}\right)^{n+1}.
$$
\end{cor}

\noindent
(The first terms of $pg_n$, $n\geq 2$, are $1,5,24,128,770,5190,38864,320704$.)

\begin{proof} Using Theorem~\ref{th3}, $PG(z)$ is obtained by calculating $\partial_y G(z,y)|_{y=1}$,
and the asymptotic follows from the classical singularity analysis.
\end{proof}

  \begin{table}[ht!]
    \centering
    $ \begin{array}{c|l|l|l}
      \mbox{Pattern } P& \mbox{Sequence counting $\mathcal{B}(P)$} & \mbox{OEIS \cite{sloa}}  & \mbox{Comments} \\\hline\hline

       \moo &
      1, 1, 2, 5, 15, 52, 203,  877,  4140, 21147, ... &
      \seqnum{A000110}~(\mbox{Bell})              &
      \mbox{Cor.~\ref{cor2} and Th. \ref{new_th_bij}}
      \\\hline
      \mo &
      1, 1, 2, 5, 16, 61, 272, 1385,  7936, 50521,  ...     &
      \seqnum{A000111} ~(\mbox{Euler})                 &
      \mbox{Cor.~\ref{cor21} and Th.~\ref{th6}}
      \\\hline
      \mooo  &
      1, 1, 2, 5, 16, 64, 308, 1730, 11104, 80176,  ...     &
      \seqnum{A131178}                  &
       \mbox{Cor.~\ref{cor31} and Th.~\ref{th4}}
      \\\hline
    \end{array}$
    \caption{Number of t-shelves avoiding the pattern $P$.}
    \label{results}
\end{table}

  \begin{table}[ht!]
    \centering
    $ \begin{array}{c|l|l}
      \mbox{Pattern } P& \mbox{Popularity of left children in $\mathcal{B}(P)$}  & \mbox{Comments} \\\hline\hline

       \moo &
      1, 5, 23, 109, 544, 2876, 16113, ... &

      \mbox{Corollary \ref{cor3}}
      \\\hline
       \mo &
      1, 4, 19, 94, 519, 3144, 20903, 151418,  ...     &

      \mbox{Corollary \ref{cor22}}
      \\\hline
      \mooo  &
      1,5,24,128,770,5190,38864,320704,  ...     &

       \mbox{Corollary \ref{cor32}}
      \\\hline
    \end{array}$
    \caption{Popularity of left children among t-shelves avoiding the pattern $P$.
      None of these sequences is yet recorded in OEIS \cite{sloa}.}
    \label{results2}
\end{table}

\section{Constructive bijections}
\label{bijections}

The counting sequences for t-shelves avoiding a pattern of length $3$
given in Corollaries \ref{cor2}, \ref{cor21} and \ref{cor31}
are known (see Table \ref{results}), and these results deserve bijective proofs.
Here, for each pattern $P\in\{\moo\,,\mo\,,\mooo\,\}$, we give an explicit
bijection between $\mathcal{B}(P)$ and a simpler combinatorial class.
These results are stated in the next three theorems, the first two of them are straightforward.

\begin{thm}
\label{new_th_bij}
There is a bijection between the set of partitions of $\{1,2,\ldots,n\}$ and
the set $\mathcal{B}_n(P)$ of t-shelves of size $n$ avoiding the pattern $P=\moo$.
\end{thm}
\proof
For a partition $S_1,S_2,\ldots,S_k$ of a set $S\subseteq\{1,2,\ldots,n\}$ with $\min S_1<\min S_2<\ldots<\min S_k$
we define a t-shelf $T$ with nodes labeled by integers from $S$. If $k=1$, then $T$ is simply the t-shelf with
no right children (and with labels in $S=S_1$). Elsewhere, $T$ is defined recursively as:
\begin{itemize}[topsep=-0.5mm,itemsep=-1mm]
\item[$-$] the root of $T$ is labeled by $\min S_1$,
\item[$-$] the left t-shelf of $T$ has size equal to $\mathrm{card}\, S_1-1$ and
does not have a right children; its nodes are labeled by integers in $S_1\setminus\{\min S_1\}$,
\item[$-$] the right t-shelf of $T$ is obtained recursively from the partition $S_2,\ldots,S_k$ of
$S\setminus S_1$.
\end{itemize}
Clearly, the t-shelf $T$ corresponding to a set partition of $\{1,2,\ldots,n\}$ is
a size $n$ t-shelf avoiding $P$.
See the recursive definition of $\mathcal{B}(P)$ in the proof of Theorem \ref{th1} and the shape of $T$ given in Figure \ref{bells}. This construction is reversible, and the statement holds
\endproof

{\it Unordered binary  increasing trees} are the non-ordered counterpart of t-shelves:
in an unordered binary increasing tree
the sibling nodes are not longer ordered among themselves, and nodes with no sibling are
not `labeled' left/right.
Thus each unordered binary increasing tree
$T$ can be expanded into $2^k$ different t-shelves of same size, where $k$ is the number of nodes of $T$ having at least one child. Theorem \ref{th6} below establishes a bijection between size $n+1$ unordered binary increasing trees
and size $n$  t-shelves avoiding $\mo$.
An interesting intermediate ordered/unordered combinatorial class is that of binary increasing trees
where, as above, the sibling nodes are not ordered, but nodes with no sibling are still `labeled' left/right. We denote by $\mathcal{J}$ the set of these trees.
%Each tree $T$ in $\mathcal{J}$ can be expanded into $2^k$ different t-shelves of same size,
%where $k$ is the number of nodes of $T$ having
%two children.

%In~\cite{kuznetsov1994increasing} and discussed in~\cite{callan2009note}.
%Before proving these two
%theorems we briefly discuss a connection between treeshelfs and
%alternating permutations, and give some useful definitions and
%constructions.

\medskip

We define a transformation $\phi$ acting on unordered binary increasing trees and on trees in $\mathcal{J}$
by ordering the nodes having a sibling: if a node of a tree has two children, then we consider the
child with the smaller label as the right one (and thus, that with the larger label as the left one).
This configuration is depicted below.

$$
  \raisebox{-0.37\height}{%
  \begin{tikzpicture}[scale = 0.4, inner sep=1pt]

    \node (n1) at (1,1) {$y$};
    \node (n2) at (2,3) {$u$};
    \node (n3) at (3,2) {$z$};

    \path
    (n2) edge (n1) (n2) edge (n3);
  \end{tikzpicture}}, \; \text{when } z < y.$$

\medskip

Clearly, applying $\phi$ on a tree in $\mathcal{J}$ a t-shelf avoiding $\mooo$ is obtained.
Moreover, this transformation is reversible, and since $\mathcal{J}$ is counted by the sequence
\seqnum{A131178} in \cite{sloa} (see the references therein),
the next theorem gives a constructive proof for the counting sequence of t-shelves avoiding $\mooo$.

\begin{thm}
\label{th4}
There is a bijection between the set $\mathcal{J}$ and the set $\mathcal{B}(P)$ of t-shelves
avoiding the pattern $P=\mooo$.
\end{thm}

\bigskip

In order to obtain a bijection between unordered binary increasing trees and t-shelves that avoid $\mo$ (next theorem),
we apply the shift (defined below) on unordered binary increasing trees in standard representation.
The {\it standard representation} of such a tree is the t-shelf obtained after ordering sibling nodes,
which is obtained by performing the above transformation $\phi$, together with considering as right child
each node with no sibling, as depicted below.
  $$
  \raisebox{-0.37\height}
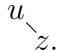
{%
  \begin{tikzpicture}[scale = 0.4, inner sep=1pt]

    \node (n1) at (1,1) {$u$};
    \node (n2) at (2,0) {$z.$};

    \path
    (n2) edge (n1);
  \end{tikzpicture}}$$

The {\it shift} of a node $y$ of a t-shelf has effect only if the following conditions are satisfied:
(i) $y$ is a left child and it has a right sibling, say $z$; and
(ii) $z$ in turn does not have a left child and its label is smaller than that of $y$.
Otherwise the shift has no effect.
With this notation, the shift of a node $y$
satisfying the two conditions above consists of pruning $y$ from its parent and grafting it as the left child of $z$, see Figure \ref{shi}.

\begin{figure}[H]
 \centering \raisebox{-0.37\height}{%
  \begin{tikzpicture}[scale = 0.6, inner sep=1pt]
    \node (n1) at (1,1) {$y$};
    \node (n2) at (2,3) {$\bullet$};
    \node (n3) at (3,2) {$z$};

    \path
    (n2) edge (n1) (n2) edge (n3);
  \end{tikzpicture}}
  $\xmapsto{\text{shift}}$
  \raisebox{-0.37\height}{%
  \begin{tikzpicture}[scale = 0.6, inner sep=1pt]
    \node (n1) at (1,1) {$y$};
    \node (n2) at (2,3) {$\bullet$};
    \node (n3) at (3,2) {$z$};

    \path
    (n2) edge (n3) (n1) edge (n3);
  \end{tikzpicture}}
\caption{The shift operation. The label of $z$ is smaller than that of $y$.}
\label{shi}
\end{figure}
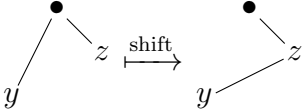

Finally, the shift of a t-shelf $T$ is defined recursively by shifting, in order,
the right t-shelf of $T$, the root of $T$, and then the left t-shelf of $T$.
See the first part of Figure~\ref{proc} for an illustration.
Obviously, the shift of a t-shelf $T$ is still a t-shelf, and if $T$ is the
standard representation of some unordered binary increasing  tree, then the shift of $T$ avoids $\mo$,
and its root does not have a left child.

\begin{figure}[ht!]
  \centering
  \raisebox{-0.5\height}{%
    \begin{tikzpicture}[scale = 0.6, inner sep=1pt]

      \node (n0) at (4,6) {$1$};
      \node (n1) at (5,5) {$2$};
      \node (n2) at (6,4) {$3$};
      \node (n3) at (2,3) {$4$};
      \node (n4) at (3,2) {$5$};
      \node (n5) at (1,1) {$6$};

      \path
      (n0) edge (n1) edge (n3)
      (n1) edge (n2)
      (n3) edge (n4) edge (n5)
      ;
    \end{tikzpicture}}\qquad
  $\xmapsto{\text{recursive shift}}$
  \raisebox{-0.5\height}{%
    \begin{tikzpicture}[scale = 0.6, inner sep=1pt]

      \node (n0) at (4,6) {$1$};
      \node (n1) at (5,5) {$2$};
      \node (n2) at (6,4) {$3$};
      \node (n3) at (2,3) {$4$};
      \node (n4) at (3,2) {$5$};
      \node (n5) at (1,1) {$6$};

      \path
      (n0) edge (n1)
      (n1) edge (n2) edge (n3)
      (n3) edge (n4)
      (n4) edge (n5)
      ;
    \end{tikzpicture}}
  $\xmapsto{\text{Root deletion}}$
  \raisebox{-0.5\height}{%
    \begin{tikzpicture}[scale = 0.6, inner sep=1pt]

     % \node (n0) at (4,6) {$1$};
      \node (n1) at (5,5) {$1$};
      \node (n2) at (6,4) {$2$};
      \node (n3) at (2,3) {$3$};
      \node (n4) at (3,2) {$4$};
      \node (n5) at (1,1) {$5$};

      \path
     % (n0) edge (n1)
      (n1) edge (n2) edge (n3)
      (n3) edge (n4)
      (n4) edge (n5)
      ;
    \end{tikzpicture}}
  \caption{A unordered binary increasing  tree in standard representation and its
    image after the recursive shift process, and after the deletion of the root.}
  \label{proc}
\end{figure}
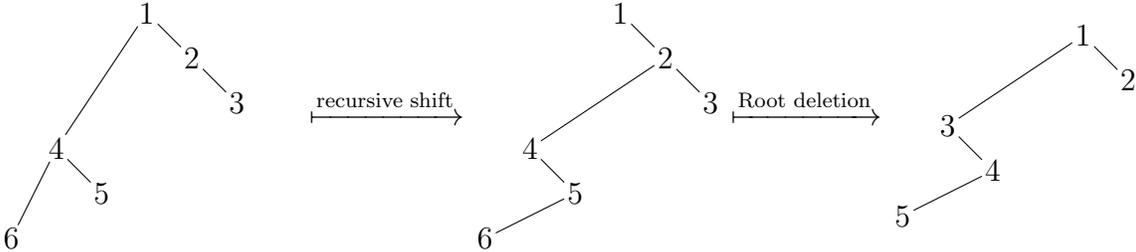

\begin{thm}
\label{th6}
There is a bijection between unordered binary increasing trees with $n+1$ nodes and the set
$\mathcal{B}_n(P)$ of t-shelves of size $n$ avoiding the pattern $P=\mo$.
\end{thm}
\begin{proof}
Let $S$ be an unordered binary increasing tree with $n+1$ nodes and $T$ be the shift of its standard representation.
As mentioned above, $T$ is a t-shelf and its root does not have a left child.
We define the mapping $S\mapsto U$, where $U$ is the t-shelf obtained after deleting
the root of $T$ and decreasing by one each label of the obtained t-shelf,
see Figure~\ref{proc}. This mapping is injective, and any t-shelf with $n$ nodes that avoids $\mo$ can be obtained by this mapping from an unordered binary increasing tree.
\end{proof}
%\section{Acknowledgements} I would like to thank the anonymous referees for their very careful reading of this
%paper and their helpful comments and suggestions.
%f\bigodotor its constructive comments.

Let us remark that unordered binary increasing trees are equinumerous with alternating permutations
that starts by a descent, as proved by Foata and
Sch{\"u}tzenberger~\cite{foa}.  The corresponding
bijection is given by Donaghey~\cite{don}. Using  Donaghey's bijection together with the bijection in Theorem~\ref{th6},
we obtain a one-to-one correspondence between  alternating
permutations starting with a descent and t-shelves avoiding $\mo.$

\bibliographystyle{amsplain}

\begin{thebibliography}{10}


\bibitem{ber} F. Bergeron, P. Flajolet and B. Salvy, Varieties of increasing trees,
{\it Proceedings of CAAP'92, Lect. Notes in Comp. Sc}. {\bf 581} (1992), 24--48.


%\bibitem{bar} J.-L. Baril, Gray code for permutations with a fixed number of cycles, {\it Discrete Math.} {\bf 307}(2007), 1559--1571.

\bibitem{bod} O. Bodini and A. Genitrini, Cuts in increasing trees, {\it Proceedings of the Meeting on Analytic Algorithmics and Combinatorics} (2015),  66-77.

\bibitem{cal} D. Callan, A note on downup permutations and increasing 0-1-2 trees, {\it preprint} at \url{http://www.stat.wisc.edu/~callan/notes/}.

\bibitem{com} L. Comtet, Advanced Combinatorics, Reidel, Dordrecht, 1974.

\bibitem{Corteel_and_co}
S. Corteel, M. A. Martinez, C. D. Savage, and M. Weselcouch,
Patterns in inversion sequences I,
{\it Discrete Mathematics and Theoretical Computer Science}
{\bf 18}(2) (2016), electronic.

\bibitem{don} R. Donaghey, Alternating permutations and binary increasing trees, {\it J. of Combinatorial Theory, Series A} {\bf 18} (1975), 141-148.

\bibitem{euler} L. Euler, De serie Lambertina Plurimisque eius insignibus proprietatibus,
{\it Acta Acad. Scient. Petropol.} {\bf 2} (1783), 29-51. Also in
{\it Opera Omnia, Series Prima}, vol. 6: Commentationes Algebraicae, Leipzig, Teubner, 1921, 350-369.


\bibitem{fla} P. Flajolet and R. Sedgewick, Analytic combinatorics, {\it Cambridge University Press}, 2000.

%\bibitem{bon1} M. B{\'o}na, Combinatorics of permutations, {\it Discrete Mathematics and its Applications}, CRC Press, Boca Raton, FL, second edition, 2012.

\bibitem{foa} D. Foata and M.P. Sch{\"u}tzenberger, Nombres d'Euler et permutations alternantes, {\it A survey of combinatorial theory} (1973), 173--187.

\bibitem{fra} J. Fran{\c c}on, Arbres binaires de recherche : propri{\'e}t{\'e}s combinatoires et applications, {\it Revue fran{\c c}aise d'automatique informatique en recherche op{\'e}rationnelle} {\bf 10} (1976), 35-50.

\bibitem{gou} I.P. Goulden and D.M. Jackson, Combinatorial enumeration, John Wiley, New York, 1983.

\bibitem{kla} M. Klazar, Counting set systems by weight, {\it Electron. J. Combin.} {\bf 12} (2005), \# R11.

\bibitem{knu73}
D.E. Knuth, The art of computer programming. Volume 3. Sorting and Searching,
Addison-Wesley Publishing Co., Reading, Mass.-London-Don Mills, Ont., 1973.

\bibitem{knu} D.E. Knuth, The Art of Computer Programming Vol. {\bf 4}, Fascicle 3, Pearson Education, 2005.

\bibitem{kut} A.G. Kutznetsov, I.M. Pak and A.E. Postnikov, Increasing trees and alternating permutations, {\it Russian Math. Surveys} {\bf 49}(6) (1994), 79-114.

\bibitem{Mansour_Shattuck}
T. Mansour and M. Shattuck, Pattern avoidance in inversion sequences,
{\it Pure Mathematics and Applications}, {\bf 25}(2) (2015) 157-176.

\bibitem{odl} A.M. Odlyzko, Asymptotic Enumeration Methods, In Handbook of Combinatorics,
Vol. {\bf 2} (1995), Cambridge, MA: MIT Press,  1063-1229.

\bibitem{pet} T.K. Petersen, Eulerian numbers, Birkh{\" a}user Advanced Texts, Springer, 2015.

\bibitem{Polya_Szego}
G. P\'olya, G. Szeg\H{o}, Aufgaben und Lehrsätze der Analysis, Berlin, Springer-Verlag, 1925.

\bibitem{sloa} N.J.A. Sloane: The On-line Encyclopedia of Integer
Sequences, available electronically at \url{http://oeis.org}.

\bibitem{sta} R.P. Stanley, {\it Enumerative Combinatorics} Vol. {\bf 2}, Cambridge University Press, 1999.



\end{thebibliography}

\end{document}